\documentclass{article}
\usepackage{amsmath,amssymb,graphicx,setspace,amsthm,float,url}
\usepackage{subcaption}
\usepackage{mathtools,blindtext, geometry}
\usepackage{tabularx}
\usepackage{xcolor}
\usepackage[hidelinks]{hyperref}
\usepackage{algorithm}
\usepackage{algpseudocode}
\usepackage{tikz}  
\usepackage[affil-it]{authblk}
\usepackage[numbers,square,sort&compress]{natbib}

\theoremstyle{definition}
\newtheorem{definition}{Definition}[section]

\newtheorem{remark}[definition]{Remark}

\theoremstyle{plain}
\newtheorem{assumption}{Assumption}[section]
\newtheorem{theorem}[definition]{Theorem}
\newtheorem{proposition}[definition]{Proposition}
\newtheorem{lemma}[definition]{Lemma}

\title{\huge A note on diffusive/random-walk behaviour in Metropolis--Hastings algorithms}

\author{Yuxin Liu$^{*, \dagger}$}
\author{Peiyi Zhou$^{*, \ddagger}$}
\author{Samuel Livingstone$^{\S}$}

\affil{Department of Statistical Science, University College London}

\begin{document}

\maketitle

\insert\footins{\noindent\footnotesize $^*$Equal contribution.}
\insert\footins{\noindent\footnotesize $^\dagger$\texttt{y.liu.24@ucl.ac.uk}; Corresponding author}
\insert\footins{\noindent\footnotesize $^\ddagger$\texttt{peiyi.zhou.25@ucl.ac.uk}; Corresponding author}
\insert\footins{\noindent\footnotesize $^\S$\texttt{samuel.livingstone@ucl.ac.uk}}

\begin{abstract}
We prove a general result that if a Metropolis--Hastings algorithm has a proposal that is not geometrically ergodic and the acceptance rate approaches unity at a suitable rate as the state variable becomes large, then the Metropolised chain will also not be geometrically ergodic. Our conditions seem stronger than might be expected, but are shown to be necessary through a counterexample. We then turn our attention to the random walk and guided walk Metropolis algorithms. We show that if the target distribution has polynomial tails the latter converges at twice the polynomial rate of the former, but that if instead the target distribution has strictly convex potential then the random walk Metropolis behaves as a $1/2$-lazy version of the guided walk Metropolis when the state variable is large, and therefore moves at a similar (ballistic) speed.
\end{abstract}


\section{Introduction}

Sampling algorithms and Monte Carlo methods based on Markov processes are an indispensable tool in modern computational research.  A prominent example is the Metropolis--Hastings algorithm, a meta-algorithm from which many popular approaches to Markov chain Monte Carlo can be derived, either as a particular example or in some appropriate limit.  The most basic form of Metropolis--Hastings algorithm will produce a $\pi$-reversible Markov chain, where $\pi$ is the target distribution from which samples are desired, but non-reversible algorithms can straightforwardly be constructed by combining Metropolis--Hastings steps with others in a cyclical manner, leading to a rich suite of methods that has captured the attention of both researchers and practitioners for many decades.

A key question of interest is how to ensure that Markov chain algorithms mix quickly, so that samples and ergodic averages are representative of the target distribution $\pi$.  
Chains that are said to exhibit `random walk' or `diffusive' behaviour take small, directionless steps and uncover little about $\pi$. Such chains are generally shunned because they often mix slowly.  A common approach to reducing random walk behaviour is to augment the state variable with a momentum, typically creating a non-reversible algorithm.  It is known that ergodic averages constructed from a non-reversible Markov chain $P$ will have at least as small an asymptotic variance as a reversible chain formed as $(P + P^*)/2$, where $P^*$ is the $L^2(\pi)$-adjoint of $P$, but the degree of improvement is dependent on both $P$ and $\pi$ \cite{andrieu2016random}.  Mixing times can also be reduced by introducing momentum, but also increased arbitrarily if care is not taken \cite{diaconis2000analysis}.

Diffusive behaviour more generally is an extensively studied phenomenon \citep{metzler2000random}.  A simple categorisation for a stochastic process $(X_t)_{t \geq 0}$ on $\mathbb{R}^d$ is to consider the displacement from its initial value $X_0 := x$ for suitably small $t$.  Typically 
\[
\mathbb{E}\|X_t - x\| \sim t^\alpha,
\]
for some $\alpha > 0$.  The case $\alpha = 1/2$ indicates \textit{diffusive} behaviour, which is characteristic of Brownian motion. If $\alpha < 1/2$ then the process is called \textit{subdiffusive}, and if $\alpha > 1/2$ it is \textit{superdiffusive}.  In the special case $\alpha = 1$ it is called \textit{ballistic}. In any instance that $\alpha \neq 1/2$ the process is said to exhibit \textit{anomalous diffusion}, which is an intense area of mathematical study \citep{metzler2000random, bouchaud1990anomalous}.

A reversible Markov chain Monte Carlo sampling algorithm can behave in a diffusive manner if the transition kernel places most of its mass in a small neighbourhood $N(x)$ of the current state $x$, and the target distribution formally satisfies $\pi(dx) \approx \pi(dy)$ for any $y \in N(x)$. In this case the detailed balance equations dictate that $P(x,dy) \approx P(y,dx)$, meaning the chain will behave in a diffusive manner when this property is combined with the locality of the transition.  A non-reversible algorithm no longer satisfies the detailed balance equations, and so even a transition kernel that still concentrates on $N(x)$ can induce anomalous, often ballistic, behaviour in the same setting.  Reversible Markov chains need not always behave diffusively, however, as the transition is not always restricted to a small neighbourhood of the current state, and the target distribution is not always `flat' in the region that a typical transition is made.  In this case the benefits of non-reversibility are less clear.

We highlight in this article cases in which reversible Metropolis--Hastings algorithms can exhibit diffusive behaviour, and others in which they behave just as non-reversible alternatives do, the key differentiator being the form of $\pi$.  Recall that a single step of a Metropolis--Hastings algorithm consists of drawing a candidate next state for the Markov chain $Y \sim Q(X_i,\cdot)$, where $Q$ is some kernel, and setting $X_{i+1} = Y$ with probability $\alpha(X_i,Y)$, where $\alpha$ is chosen to enforce $\pi$-reversibility.  We begin in Section \ref{sec:MH_high_accept} with a very general result showing that if typical candidate next steps result in $\alpha(X_i,Y) \approx 1$ when $\|X_i\|$ is large, and $Q$ itself is not geometrically ergodic, then the Metropolis--Hastings transition $P$ will typically not be either.  This is a common scenario in which diffusive behaviour can often be observed, because $Q$ alone does not promote fast mixing and $\alpha$ does not influence the dynamics to the degree that $P$ behaves any differently to $Q$.  Surprisingly, however, the condition that $\lim_{\|x\|\to\infty}\int \alpha(x,y)Q(x,dy) = 1$ alone is not sufficient for such a result; we provide a counterexample in which it holds and $Q$ is not geometrically ergodic but $P$ in fact is.  We also discuss extensions in which $\alpha \not\to 1$ but nonetheless the distribution becomes asymptotically flat in some directions, and connect this to the concept of acceleration in sampling and optimisation \cite{oberdorster2025accelerated, bubeck2015convex}.  

In Section \ref{sec: RWM and GWM on 1D R} we focus on two specific algorithms, the random walk and guided walk Metropolis, which can be viewed as a canonical reversible Metropolis--Hastings algorithm and its natural non-reversible counterpart.  If $\pi$ has a density with polynomial tails then it is well-known that the random walk Metropolis can be diffusive while the guided walk exhibits ballistic motion.  We prove that this difference is consequential, in that the guided walk on $\mathbb{R}$ exhibits a faster polynomial rate of convergence to equilibrium.  We then consider the case in which $\pi$ has a strictly log-concave density, meaning that the tails are lighter. In this case the behaviour differs markedly; in the limit $|x|\to\infty$, the random walk Metropolis behaves as a $1/2$-lazy version of the guided walk Metropolis.  Despite $Q$ being a random walk, the acceptance rate qualitatively changes the transition resulting in very little difference in the behaviour of the reversible and non-reversible algorithms during the transient phase, with both exhibiting ballistic motion.

\section{Metropolis--Hastings \& the high-acceptance regime}
\label{sec:MH_high_accept}

The Metropolis--Hastings algorithm provides a simple recipe to construct a Markov transition kernel. Given the current state $X_i$ a candidate kernel $Q$ is used to generate $Y \sim Q(X_i,\cdot)$.  After this, the next state $X_{i+1}$ is set to be $Y$ with probability $\alpha(X_i,Y)$.  There are infinitely many possible choices for $\alpha$, but the most common when both $\pi$ and $Q$ possess densities $\pi(x)$ and $q(x,y)$ with respect to a common dominating measure is $\alpha(x,y) = \min\left( 1, r(x,y)\right)$, where 
\[
r(x,y) := \frac{\pi(y)q(y,x)}{\pi(x)q(x,y)}
\]
is called the Hastings ratio (if $\pi(x)q(x,y) = 0$ then $\alpha(x,y) := 0$).  More generally any choice satisfying $\alpha(y,x) = r(x,y)\alpha(x,y)$ and $\alpha(x,y) \in [0,1]$ for all $x,y$ results in a Metropolis--Hastings transition kernel
\[
P(x,dy) = \alpha(x,y)Q(x,dy) + \left( 1 - \int \alpha(x,y)Q(x,dy) \right) \delta_x(dy)
\]
that is $\pi$-reversible.  The choice $\alpha(x,y) := \min(1, r(x,y))$ is almost always preferred as it minimises the asymptotic variance of ergodic averages \cite{peskun1973optimum,tierney1998note}.  We restrict attention to this choice throughout (though see Remark \ref{rem:general_alpha}).

A non-reversible algorithm that is still $\pi$-invariant can straightforwardly be constructed by cycling through different Metropolis--Hastings kernels in a manner that is not palindromic.  A simple example is to consider two candidate kernels $Q_1$ and $Q_2$ and respective Metropolis--Hastings kernels $P_1$ and $P_2$.  The kernel $P := P_1P_2$ is no longer reversible since $P^* = P_2P_1 \neq P$, but it is clearly still $\pi$-invariant.  One popular approach is to augment the state space to $\xi := (x,p)$, and define $P_1$ as a Metropolis--Hastings kernel with $Q_1(\xi, d\xi') := \delta_{T(\xi)}(d\xi')$ with $T$ an involution (meaning $T^{-1} = T$).  The kernel $P_2$ then only updates the \textit{momentum} $p$, leaving the \textit{position} $x$ unchanged.  Extensions of Peskun's theorem to non-reversible chains are given in \citep{maire2014comparison,andrieu2021peskun}.

Throughout we assume that $\pi$ is defined on $\mathbb{R}^d$.  Recall that a function $f: \mathbb{R}^d \to (0,\infty)$ is called \textit{lower semicontinuous} if for each $c \in \mathbb{R}$, the set $\{x \in \mathbb{R}^d : f(x) \le c\}$ is closed, or equivalently, $x_n\to x$ implies $\liminf_{n\to\infty} f(x_n) \ge f(x)$ (e.g. \cite[Section 2.1]{aliprantis2006infinite}). A function $g:\mathbb{R}^d \to (0,\infty)$ is called \textit{upper semicontinuous} if and only if $1/g(x)$ is lower semicontinuous.  A Markov transition kernel $P$ is called a $T$-chain if there is a distribution $\nu$ on $\mathbb{N}$ and a subtransition kernel $T$ such that $\int P^n(x,A)\nu(dn) \geq T(x,A)$, $T(x,\mathbb{R}^d) > 0$ for all $x$ and the function $f_A(x) := T(x,A)$ is lower semicontinuous for all $A \in \mathcal{B}(\mathbb{R}^d)$. In many cases $\nu = \delta_1$ is a reasonable choice.  The concept of a $T$-chain is introduced in Chapter 6 of \citep{meyn2012markov}, and allows straightforward characterisation of many topological and stability properties of $P$.  

A Markov chain with transition $P$ will be called geometrically ergodic if there is a $\pi$-a.e. finite $V \geq 1$, an $R < \infty$ and $\rho < 1$ such that for all $n \geq 1$ and any $x$ for which $V(x) < \infty$
\begin{equation}
\label{eq:geometric_ergodicity}
\|P^n(x,\cdot) - \pi\|_{V} \leq RV(x)\rho^n,
\end{equation}
where for a signed measure $\mu$ the $V$-norm $\|\mu\|_V := \sup_{|g|\leq V}\left|\int g(x)\mu(dx)\right|$.  It is well-known that among $\pi$-reversible chains geometric ergodicity is equivalent to the existence of a positive $L^2(\pi)$ spectral gap \citep{kontoyiannis2012geometric}.  The theory allows for $V$ to be infinite on a $\pi$-null set, however if an $x$ is chosen such that $V(x) = \infty$ then the bound ceases to be useful.  We therefore use \textit{geometric ergodicity} to refer exclusively to the case in which $V$ is finite for all $x$ in the support of $\pi$; this allows the useful equivalence in Proposition \ref{prop: MH is a T-chain}(ii), following Proposition 3.1 of \citep{roberts1996geometric}.

There are two classical settings in which the Metropolis--Hastings algorithm is well-known to perform poorly. The first is when $\alpha(x,y) \approx 0$ in some region of the state space, meaning almost all proposals are rejected.  It is shown in \cite{roberts1996geometric} that geometric ergodicity cannot then hold.  The other is when $\alpha(x,y) \approx 1$ for almost all proposals $y$ and $x$ sufficiently large, and $Q$ itself is not geometrically ergodic but meanders aimlessly.  This is the canonical `random walk behaviour' scenario, yet to the best of our knowledge no result directly characterising its relation to geometric ergodicity exists (though related results appear in \cite{jarner2003necessary}).  We provide one as Theorem \ref{them:non-geometric-ergodicity}, under mild regularity conditions on $\pi$ and $Q$ stated below.

The following assumption facilitates Proposition \ref{prop: MH is a T-chain} below, a generalisation of Theorem 2.2 in \cite{roberts1996geometric}, from which precise statements about $P$ will follow.


\begin{assumption}
    The distribution $\pi$ is defined on $\mathbb{R}^d$ and the following hold
    \begin{itemize}
        \item[(i)] $\pi$ has a density $\pi(x)$ that is upper semicontinuous and bounded away from 0 and $\infty$ on compact sets.
        \item[(ii)] $Q(x,\cdot)$ has a density $q(x,y)$, both $q(x,\cdot)$ and $q(\cdot,y)$ are lower semicontinuous for all $x$ and $y$, and there exists $\delta_q > 0$ such that for every $x \in \mathbb{R}^d$ \[
    \|x-y\| \leq \delta_q \implies q(x,y) > 0.
    \]
    \end{itemize}
    \label{ass:pi_q_regularities}
\end{assumption}

\begin{proposition}
\label{prop: MH is a T-chain}
    Let $\{X_n\}_{n\geq 0}$ be a Metropolis--Hastings Markov chain on $\mathbb{R}^d$ such that $\pi$ and $Q$ satisfy Assumption \ref{ass:pi_q_regularities}.  Then
    \begin{itemize}
        \item [(i)] $\{X_n\}_{n\geq 0}$ is a $\pi$-irreducible and aperiodic $T$-chain, and therefore all compact sets are small.
        \item [(ii)]  It is geometrically ergodic if and only if there exists a function $V: \mathbb{R}^d\to[1,\infty)$ such that 
        $$
        \limsup_{\|x\|\to\infty} \frac{PV(x)}{V(x)}<1.
        $$
    \end{itemize}
\end{proposition}

\begin{proof}
See Appendix \ref{app:proofs_MH_high_accept}.
\end{proof}

Our main result of this section is stated as Theorem \ref{them:non-geometric-ergodicity} below.  Intuitively, one might expect that if the acceptance rate of a Metropolis--Hastings algorithm approaches 1 in the tail region, then the transition kernel $P$ will behave similarly to $Q$ in this region. The tail behaviour of $Q$ should then determine the convergence rate of the algorithm. Consequently, if $Q$ fails to be geometrically ergodic and the acceptance rate approaches 1 in the tail, then $P$ should also fail to be geometrically ergodic.  We show below that this reasoning is essentially correct, but that a slightly stronger condition is required (condition (ii) below) than might be expected on the rate at which the acceptance rate tends to unity as $\|x\|\to \infty$.

\begin{theorem}
\label{them:non-geometric-ergodicity}
Let $\pi$ and $Q$ satisfy Assumption \ref{ass:pi_q_regularities}, and suppose in addition
\begin{enumerate}
    \item[(i)] $Q$ is not geometrically ergodic.
    \item[(ii)] For any $V : \mathbb{R}^d \to [1,\infty)$ such that $\int V(x)\pi(dx) < \infty$ it holds that 
    \[
    \lim_{\|x\|\to\infty} \int \frac{V(y)}{V(x)}(\alpha(x,y) - 1) Q(x,dy) = 0.
    \]
\end{enumerate}
Then the Metropolis--Hastings kernel $P$ is not geometrically ergodic.
\end{theorem}

\begin{proof}
First note that for any suitable $V$
\begin{align}
\label{eq:pv_expression}
    \frac{PV(x)}{V(x)}
    &= \int \frac{V(y)}{V(x)} \alpha(x,y)Q(x,dy) + \int \frac{V(y)}{V(x)} \delta_x(dy)\left(1- \int\alpha(x,y) Q(x,dy)\right)\nonumber\\
    &= \int \frac{V(y)}{V(x)} Q(x,dy) +\left(1- \int\alpha(x,y) Q(x,dy)\right)+D(x)
\end{align}
where $D(x) := \int V(y)(\alpha(x,y)-1)Q(x,dy)/V(x)$.  Under Assumption \ref{ass:pi_q_regularities} the extension of Proposition \ref{prop: MH is a T-chain}(ii) to the Markov chain produced by iterating $Q$ is immediate.   Condition (i) of the theorem statement therefore implies
\[
\limsup_{\|x\|\to\infty}\int \frac{V(y)}{V(x)}Q(x,dy) \geq 1.
\]
The second term in \eqref{eq:pv_expression} is non-negative (although by assumption it asymptotes to 0), and by condition (ii) of the theorem statement $\lim_{\|x\|\to\infty}D(x) = 0$. Combining gives
\begin{equation}
    \limsup_{\|x\|\to\infty}\frac{PV(x)}{V(x)} \geq \limsup_{\|x\|\to\infty}\int\frac{V(y)}{V(x)}Q(x,dy) \geq 1
\end{equation}
From this it is clear that $P$ cannot be geometrically ergodic by applying Proposition \ref{prop: MH is a T-chain}(ii).
\end{proof}

\begin{remark}
    When $Q$ satisfies a random walk type condition, similar results follow from \cite{jarner2003necessary}. Theorem \ref{them:non-geometric-ergodicity} also covers cases where this fails, such as the schemes in \cite{livingstone2021geometric,girolami2011riemann,roberts2002langevin,xifara2014langevin}.
\end{remark}

\begin{remark} \label{rem:general_alpha}
    The above can be straightforwardly extended to the setting in which $\alpha(x,y)$ is not chosen to be the $\min(1, r(x,y))$, provided that certain regularity conditions are still satisfied.  It may then be that $\lim_{\|x\|\to\infty}\int \alpha(x,y)Q(x,dy)  = c$ for some $c < 1$, which results in $P$ being a $(1-c)$-lazy version of $Q$ in the tails.
\end{remark}

It is tempting to conjecture that if $Q$ is not geometrically ergodic and $\lim_{\|x\|\to\infty} \int \alpha(x,y)\,Q(x,dy) = 1$ then $P$ will not be geometrically ergodic. This would provide a weaker sufficient criterion than Theorem \ref{them:non-geometric-ergodicity}(ii), since $\lim_{\|x\|\to\infty} \int \alpha(x,y)\,Q(x,dy) = 1$ is a special case of $\lim_{\|x\|\to\infty} \int V(y)/V(x)[\alpha(x,y)-1]\,Q(x,dy) = 0$ with $V(x) \equiv 1$. This conjecture is false, however, as demonstrated by the following counterexample.

\begin{proposition}
\label{prop: counterexample}
    Take $\pi(x)\propto \exp(-x^2)$ and set
    $$
    Q(x,dy) = [1- \varepsilon(x)] Q_1(x,dy)+\varepsilon(x) Q_2(x,dy),
    $$
    where $\varepsilon(x) = 1/(2+|x|)$, $Q_1(x,\cdot)$ denotes a $N(x/2,1)$ distribution, and $Q_2(x,\cdot)$ denotes a uniform distribution on $(\exp(x^2),\exp(x^2)+1)$.  Then $\lim_{|x|\to\infty} \int \alpha(x,y) Q(x,dy) = 1$ and $Q$ is not geometrically ergodic, but the Metropolis--Hastings kernel $P$ \emph{is} geometrically ergodic.
\end{proposition}

\begin{proof}
See Appendix \ref{app:proofs_MH_high_accept}.
\end{proof}

The intuitive explanation for Proposition \ref{prop: counterexample} is as follows. The candidate kernel $Q(x,dy)$ is a mixture of a well-behaved distribution $Q_1(x,dy)$ with probability $1-\varepsilon(x)$ and a large jump that occurs with probability $\varepsilon(x)$. In the tail region, nearly all proposals from $Q_1(x,dy)$ are accepted, while all large jumps are rejected. Since $\lim_{|x|\to\infty} \varepsilon(x) = 0$, almost all proposals in the tail come from $Q_1(x,dy)$, causing the average acceptance rate to approach 1. Although the large jump has low probability, when it occurs $V(y)$ becomes extremely large (where $y$ is the proposed state). This prevents $Q$ from being geometrically ergodic. The Metropolis--Hastings algorithm, however, rejects nearly all such large jumps, making the transition kernel $P$ behave as $Q_1$, which is geometrically ergodic.

A natural extension of Theorem \ref{them:non-geometric-ergodicity} is to consider the case in which $\pi$ has a density that becomes asymptotically flat in only some directions.  This is common in statistical models in the presence of near non-identifiability. As a simple example let $\pi(x) = \pi_1(x_1)\pi_2(x_2)$ and $Q(x,dy) := Q_1(x_1,dy_1)Q_2(x_2,dy_2)$ with a symmetric transition density.  If $\pi_1$ has polynomial tails then taking the sequence $x^{(n)} := (n, x_2)$ for $x_2$ fixed it holds that 
\[
\lim_{n\to\infty}\int \alpha(x^{(n)}, y)Q(x^{(n)}, dy) = \int \min\left(1, \frac{\pi_2(y_2)}{\pi_2(x_2)}\right)Q_2(x_2, dy_2).
\]
We do not therefore have that $\alpha \to 1$, but rather that it becomes asymptotically independent of $x_1$ in some directions. The $x_1$-coordinate process will then behave as an $x_2$-modulated lazy version of $Q_1$ in some areas of $\mathbb{R}^2$.  When $Q$ is a random walk kernel, for example, it can be shown that such chains will not be geometrically ergodic using the results of \cite{jarner2003necessary}.  

The previous paragraph connects with the popular idea of acceleration in sampling and optimisation. In the sampling context, take $\pi$ to have an $m$-strongly log-concave density $\pi(x)$ for which $\nabla \log\pi(x)$ is gradient-Lipschitz with constant $M$. A condition number $\kappa := M/m \geq 1$ can then be associated with $\pi$, and often connected with the mixing time of a Markov chain algorithm targeting $\pi$. Often the mixing time is linear in $\kappa$, but for some \textit{accelerated} algorithms it is reduced to $\sqrt{\kappa}$.  Large $\kappa$ typically means $\pi$ is flatter in some directions than others, and mixing is determined by how well these flatter directions are explored. The square root improvement follows directly from the fact that accelerated chains move ballistically rather than diffusively in these directions (see e.g. \cite{oberdorster2025accelerated}).  Linear pre-conditioning can reduce the condition number and thereby increase sampling efficiency \cite{hird2025quantifying}. Similarly, applying a nonlinear transformation to induce faster-decaying tails can recover geometric ergodicity \cite{johnson2012variable}.

\section{Random \& Guided Walk Metropolis on $\mathbb{R}$}
\label{sec: RWM and GWM on 1D R}

We now turn to the random walk and guided walk Metropolis on $\mathbb{R}$. The (Gaussian) random walk Metropolis on $\mathbb{R}$ is the case of Metropolis--Hastings in which $Q(x,\cdot)$ is $N(x, \epsilon^2)$. The candidate density $q(x,y)$ is then symmetric in its arguments, so the Hastings ratio reduces to $r(x,y) = \pi(y)/\pi(x)$. The algorithm is well-known and well-studied (e.g. \cite{sherlock2010random}).

The guided walk Metropolis was introduced in \cite{GWM_intro} as a simple way to incorporate momentum into a sampling algorithm.  The state space is augmented to be $\mathbb{R} \times \{-1,+1\}$ with the new state variable $\xi := (x,p)$ such that $p \in \{-1,+1\}$ enforces a direction to the transitions of the Markov chain.  The algorithm produces a Markov transition from the current state $(X_i,P_i)$ by generating the proposal
\[
Y = X_i + \epsilon P_i|Z|
\]
for some $Z \sim N(0,1)$, and then setting 
\[
(X_{i+1},P_{i+1}) 
= 
\begin{cases}
    (Y,P_i), & \text{with probability } \alpha(X_i,Y) \\
    (X_i,-P_i), & \text{otherwise},
\end{cases}
\]
where $\alpha(X_i,Y) = \min(1, \pi(Y)/\pi(X_i))$ as in the random walk Metropolis.  Iterating produces a Markov chain on $\mathbb{R} \times \{-1,+1\}$ with invariant distribution $\mu(d\xi) = \pi(dx)\nu(dp)$, where $\nu$ denotes the uniform distribution on $\{-1,+1\}$.  The guided walk transition can be viewed as the two-skeleton of a two-cycle Metropolis--Hastings kernel $P_1P_2$, where $P_1$ has candidate kernel $Q_1(\xi,d\xi') = \int \delta_{(x+\epsilon p|z|,-p)}(dx',dp')\phi(z)dz$, with $\phi$ denoting the standard Normal density on $\mathbb{R}$, and $Q_2$ consists of keeping $x$ fixed and setting $p \to -p$, a candidate transition that will always be accepted owing to the symmetry of $\nu$.  Further discussion on the general construction of such deterministic transitions is given in \cite{tierney1998note,andrieu2020general, glatt2023accept}.

The key advantage of the guided walk is that if a proposal $Y$ is accepted then the chain continues moving in the same direction.  If the chain is in a region for which $\pi(x)$ is relatively flat, therefore, this will promote ballistic motion when the random walk Metropolis will instead behave diffusively.  In \cite{andrieu2021peskun} it is shown that if $P$ denotes the guided walk transition then $(P+P^*)/2$ is the random walk Metropolis (or rather an augmented version on $\mathbb{R} \times \{-1,+1\}$ for which the momentum is completely refreshed at every iteration).  When a proposal is rejected, however, the direction of travel in the guided walk is changed, and so in the presence of a non-trivial $\alpha$ the degree of benefit over the random walk Metropolis is not so clear.  We study both the polynomial-tailed case, in which $\pi$ becomes flat for large $|x|$, and the strictly log-concave case, in which $-\log \pi(x)$ grows faster than linearly in $|x|$.

\subsection{Polynomial tails}
\label{subsec:heavy_tailed_target}

In this section we assume $\pi$ satisfies the following condition, which is representative of many heavy-tailed distributions on $\mathbb{R}$.
\begin{assumption}
\label{ass:heavy-tailed}
Assumption \ref{ass:pi_q_regularities}(i) holds and in addition there exist constants $r>0$, $K>0$, and $C_0>0$ such that for $|x|>K$,
$$
\pi(x) = \frac{C_0}{|x|^{1+r}}.
$$  
\end{assumption}

A Markov chain is called polynomially ergodic if there is a constant $\beta > 0$ and real-valued functions $V_\beta \geq 1$, $M > 0$  such that for all $n \geq 1$
\begin{equation}
\label{eq:polynomial_ergodicity}
\|P^n(x,\cdot) - \pi\|_{V_\beta} \leq M(x) n^{-\beta}.
\end{equation}
Note that in this case $V_\beta$ depends on $\beta$, and in fact there is an inherent trade-off between the rate of convergence and its choice \citep{jarner2007convergence}. We therefore focus on the total variation norm, i.e. choosing $V_\beta \equiv 1$.  The largest choice of $\beta$ for which \eqref{eq:polynomial_ergodicity} holds with $V_\beta \equiv 1$ is called the \textit{polynomial rate}.

Under Assumption \ref{ass:heavy-tailed} the following result is proven in \citep{jarner2007convergence} about the random walk Metropolis.

\begin{proposition}[Jarner \& Roberts, 2007]
    Under Assumption \ref{ass:heavy-tailed} the (Gaussian) random walk Metropolis is polynomially ergodic with rate $r/2$.
\end{proposition}

For the guided walk Metropolis, by contrast, we establish a faster polynomial rate of convergence.

\begin{proposition}
    Under Assumption \ref{ass:heavy-tailed} the guided walk Metropolis is polynomially ergodic with rate $r$.
\end{proposition}

\begin{proof}
    In Lemma \ref{lem:lowerbound_polynomialtails} we show that the polynomial rate is at least $r$, and in Lemma \ref{lem:upperbound_polynomialtails} we show that it is $\leq r$.  Both results can be found in Appendix \ref{app:proofs_gwm_polynomialtails}.
\end{proof}

The proof follows a similar argument to that used in \citep{vasdekis2022note} to study a related continuous-time process.  The result shows that in this setting a doubling in rate is achieved.  The guided walk Metropolis is a \textit{lifted} version of the random walk Metropolis, and it is known that when the state space is finite or compact a lifted Markov chain must have mixing time at least the square root of the original chain, which translates to at most a doubling of the rate \cite{chen1999lifting,ramanan2018bounds}.

\subsection{Strictly log-concave tails}
\label{subsec:light_tailed_target}

We now consider the case of $\pi$ having a density and write $U(x) := -\log\pi(x)$ to denote the associated \textit{potential} function.  Assume the following strict convexity and super-linear growth conditions.
\begin{assumption}
\label{ass:strict_convexity}
    The target distribution $\pi$ has a Lebesgue density $\pi(x) \propto e^{-U(x)}$ on $\mathbb{R}$ with $U \in C^1(\mathbb{R})$, $U'$ is strictly monotonically increasing and
    \begin{equation*}
    \lim_{|x| \to \infty}\frac{U(x)}{|x|} = \infty.
    \end{equation*}
\end{assumption} 
It follows straightforwardly from Assumption \ref{ass:strict_convexity} that $U'(x) \to \infty$ as $x \to \infty$ and $U'(x) \to -\infty$ as $x \to -\infty$.

Recall that in general $\check{P}$ is called an $\varepsilon$-lazy version of $P$ if $\check{P}(x,\cdot) := (1-\varepsilon)P(x,\cdot) + \varepsilon \delta_x(\cdot)$ for all $x$.  The main result of this section is the following proposition, which shows that in the typical light-tailed scenario characterised by Assumption \ref{ass:strict_convexity} the random walk Metropolis behaves just as a $1/2$-lazy version of the guided walk Metropolis for large $|x|$.  

\begin{proposition}
\label{prop:nstep_coupling}
Let $\mathcal{L}(X_n)$ and $\mathcal{L}(W_n)$ denote the laws of the $n$th states of the random walk Metropolis chain $\{X_n\}_{n \geq 0}$ and the $1/2$-lazy guided walk Metropolis chain $\{(W_n,P_n)\}_{n\geq 0}$.  If $\pi$ satisfies Assumption \ref{ass:strict_convexity}, $X_0 = W_0 = x$ and $P_0 = -\operatorname{sgn}(x)$ then for any fixed $n \in \mathbb{N}$
\begin{equation*}
\lim_{|x|\to \infty} \|\mathcal{L}(X_n) - \mathcal{L}(W_n)\|_{TV} = 0.
\end{equation*}
If instead $X_0 = W_0 = x$ and $P_0 = \operatorname{sgn}(x)$, then
\begin{equation*}
\lim_{|x|\to \infty} \|\mathcal{L}(X_{n-1}) - \mathcal{L}(W_n)\|_{TV} = 0.
\end{equation*}
\end{proposition}

\begin{proof}[Proof of Proposition \ref{prop:nstep_coupling}]
We describe the case $x \to \infty$, and simply note that the $x \to -\infty$ case follows an identical argument.  We couple the random walk and the lazy guided walk proposals synchronously in the following way. Given the starting position $X_0 = x$ for the random walk and $(W_0, P_0) = (x,-1)$ for the lazy guided walk algorithms, draw $Z \sim N(0,1)$ and set the random walk proposal to be $x' = x - \epsilon Z$. The lazy guided walk proposal can be set as $(w',p') = (x - \epsilon \max(Z,0), -1)$, since $\mathbb{P}(\max(Z,0) = 0) = 1/2$, in which case the chain will be lazy and stay put. If $Z \geq 0$ then clearly $x' = w'$, so upon acceptance $X_1 = W_1 = x'$ and $P_1 = -1$.  If $Z < 0$ then there is a chance that the chains can separate if the random walk proposal is accepted, since the lazy guided walk chain will not move.  In Lemma \ref{lem:convergence_alpha_light_tails} below, however, we show that under Assumption \ref{ass:strict_convexity} the probability of accepting such a random walk proposal tends to $0$ as $x \to \infty$.  Similarly it is easy to show that when $Z > 0$ under the same assumption $\alpha(x,x') = 1$ whenever $Z \in [0, x/\epsilon]$ for sufficiently large $x$, and that $\mathbb{P}(Z < x/\epsilon | Z > 0) \to 1$ as $x \to \infty$.  It therefore holds that $\mathbb{P}(X_1 \neq W_1) \to 0$ as $x \to \infty$.  Iterating the argument to show that $\mathbb{P}(X_n \neq W_n) \to 0$ for any fixed and finite $n$ is tedious but straightforward, since a large enough $x$ can be chosen that $\mathbb{P}\left( x- \epsilon\sum_{i=1}^n \max(0, Z_i) > c\right)$ is arbitrarily close to 1 for any fixed $c < \infty$, where each $Z_i \sim N(0,1)$ represents the Gaussian draw used for the $i$th proposal of each algorithm. We omit the details for brevity.  Applying the coupling inequality then gives $\|\mathcal{L}(X_n) - \mathcal{L}(W_n)\|_{TV} \leq \mathbb{P}(X_n \neq W_n) \to 0$ as $x \to \infty$, which completes the proof in the case $P_0 = -\operatorname{sgn}(x)$.

In the case $P_0 = 1$ then by the same Lemma \ref{lem:convergence_alpha_light_tails} the initial guided walk proposal will be rejected with probability tending to 1 as $x \to \infty$. With probability tending to 1, therefore, $(W_1,P_1) = (x, -1)$. The same coupling as above can then be constructed to connect $X_{n-1}$ and $W_n$ in law as required.  This completes the proof.
\end{proof}

\section{Discussion}

We have shown that non-reversible Metropolis--Hastings algorithms can sometimes produce much faster mixing chains than reversible counterparts, but that in other cases they behave similarly.  In particular, a random walk Metropolis will not always exhibit `random walk' behaviour; this depends crucially on the form of $\pi$, specifically the tails, which determine how much the Metropolis--Hastings chain $P$ is altered compared to the random walk proposal $Q$.  Non-reversible chains appear to be most advantageous when $\pi$ has heavy tails or more generally becomes flat in some region or direction.

Extensions of the guided walk Metropolis to $\mathbb{R}^d$ have been constructed but tend to be inferior to other (reversible) approaches \citep{gustafson2004value}.  Continuous-time processes that are similar in spirit have, however, shown more success in high dimensions \citep{bierkens2019zig}.  Another super-diffusive process is a random walk with heavy-tailed increments (see e.g. \citep{metzler2000random,klafter1987stochastic}).  Heavy-tailed increments can often be induced in sampling algorithms via randomised step-sizes, which are commonly employed and known to have benefits \cite{grazzi2026randomized}.  Using a heavy-tailed proposal in the random walk Metropolis can also increase the polynomial rate of convergence when $\pi$ has polynomial tails \cite{jarner2007convergence}.  We suspect that a similar improvement in the rate of a modified guided walk algorithm in which $Z$ is heavy-tailed can be shown, but leave the details for future work.  

\subsection*{Acknowledgments}

The authors would like to thank the Isaac Newton Institute for Mathematical Sciences, Cambridge, for support and hospitality during the programme \textit{Stochastic systems for anomalous diffusion}, where work on this paper was undertaken. This work was supported by EPSRC grant EP/Z000580/1.

\clearpage
\pagenumbering{arabic}
\renewcommand*{\thepage}{References | Page \arabic{page}}

\bibliographystyle{abbrvnat}
\bibliography{refs}

\clearpage
\pagenumbering{arabic}
\renewcommand*{\thepage}{A.\arabic{page}}
\renewcommand{\thesection}{\Alph{section}}
\appendix

\section{Proofs of Section \ref{sec:MH_high_accept}}
\label{app:proofs_MH_high_accept}

\begin{proof}[Proof of Proposition \ref{prop: MH is a T-chain}]
(i) \textit{$T$-chain:} First note that choosing $\nu = \delta_1$ gives 
\begin{equation*}
    P(x,A) \ge\int_A \alpha(x,y)q(x,y)dy =: T(x,A)
\end{equation*}
for any $A \in \mathcal{B}(\mathbb{R}^d)$. Define the function
\begin{align}
    T_y(\cdot) &:= \alpha(\cdot,y)q(\cdot,y) = \min\left\{q(\cdot,y),\frac{\pi(y)q(y,\cdot)}{\pi(\cdot)}\right\},
\end{align}
meaning $y$ is treated as a constant and $T_y(x)$ is viewed as a function of $x$ only.  By assumption $\pi(x)$ is upper semicontinuous and positive, meaning $1/\pi(x)$ is lower semicontinuous, as is the function $q(y,\cdot)$. Therefore for any sequence $x_n\to x$
$$
\liminf_{n\to\infty} \frac{q(y,x_n)}{\pi(x_n)}\ge \left( \liminf_{n\to\infty} q(y,x_n)\right) \cdot \left( \liminf_{n\to\infty} \frac{1}{\pi(x_n)}\right)\ge \frac{q(y,x)}{\pi(x)},
$$
meaning $\pi(y) q(y,\cdot)/\pi(\cdot)$ is lower semicontinuous.  The minimum of two lower semicontinuous functions is also lower semicontinuous, since for each $c \in \mathbb{R}$
$$
\{x\in \mathcal{X}:\min\{g(x),h(x)\}\le c\} =\{x\in \mathcal{X}:g(x)\le c\} \cup \{x\in \mathcal{X}:h(x)\le c\},
$$
which is closed if $g$ and $h$ are both lower semicontinuous. Combining gives that $T_y(x)$ is lower semicontinuous in $x$.  Upon noting that $T(x,A) = \int_A T_y(x) dy$, then for any sequence $x_n\to x$
\begin{align}
    \liminf_{n\to\infty} T(x_n,A) = \liminf_{n\to\infty}  \int T_y(x_n)\mathbb{I}_A(y)dy
    \ge   \int \liminf_{n\to\infty} T_y(x_n)\mathbb{I}_A(y)dy
    \ge \int_A T_y(x)dy  = T(x,A)
\end{align}
using Fatou's lemma and lower semicontinuity of $T_y(x)$, which shows that $T(\cdot, A)$ is lower semicontinuous for any $A \in \mathcal{B}(\mathbb{R}^d)$. Clearly $T(x,\mathbb{R}^d) \geq \int_{\|y-x\|\leq \delta_q}q(x,y)\alpha(x,y)dy > 0$ for any $x$ by Assumption \ref{ass:pi_q_regularities} and therefore $\{X_n\}_{n\geq 0}$ is a $T$-chain.

\textit{$\pi$-irreducibility:} 
By assumption if $\|x-y\|\le \delta_q$, then $q(x,y)>0$. This together with the assumptions on $\pi(x)$ implies
\begin{equation} \label{eq:q_alpha_lb}
q(x,y)\alpha(x,y) = \min\left(q(x,y),\frac{\pi(y)q(y,x)}{\pi(x)}\right) > 0.
\end{equation}
Consider any set $A \in \mathcal{B}(\mathbb{R}^d)$ such that $\pi(A) > 0$. Since $\pi$ has a density that is positive on all of $\mathbb{R}^d$, $A$ must have positive Lebesgue measure. The inner regularity property of the latter implies that $A$ contains a compact subset $K \subseteq A$ also of positive Lebesgue measure. For a fixed $x$ and any $z \in K$, there exists a path of finite length $n$ and points $x = x_0, x_1, \dots, x_n = z$ such that $\|x_i - x_{i-1}\| < \delta_q$. Since $K$ is bounded, we can choose an $n$ large enough to reach any $z \in K$ from $x$. The sub-transition density $p^n(x,z)$ associated with generating $z$ from $x$ in $n$ steps satisfies
\[
p^n(x, z) \geq \int \dots \int \prod_{i=1}^n q(x_{i-1}, x_i)\alpha(x_{i-1},x_i) dx_1 \dots dx_{n-1} > 0
\]
for all $z \in K$. It then follows that the transition probability to the set $A$ satisfies
\[
P^n(x, A) \ge P^n(x,K) \geq \int_K p^n(x,z)dz > 0
\]
from which it follows that the chain is $\pi$-irreducible.

\textit{Aperiodicity:}
To show aperiodicity, it suffices to prove that for all $x\in \mathbb{R}^d$, there exists $A \in \mathcal{B}(\mathbb{R}^d)$ such that both $P(x,A)>0$ and $P^2(x,A)>0$.  Let $B_x =\{y:\|y-x\|<\delta_q\}$, and note that $P(x,B_x) >0$ by Assumption \ref{ass:pi_q_regularities}.  For $P^2(x,B_x)$, we have
\begin{equation}
    P^2(x,B_x) = \int P(x,dy)P(y,B_x) \ge \int_{B_x} q(x,y)\alpha(x,y)P(y,B_x)dy.
\end{equation}
For $y\in B_x$, define $B_y  = \{z:\|z-y\|<\delta_q\}$. Since $\|y-x\| < \delta_q$, the intersection $B_x\cap B_y$ is nonempty. Therefore,
$$
P(y,B_x)\ge P(y, B_x\cap B_y) \ge \int_{B_x\cap B_y} q(y,z)\alpha(y,z)dz>0.
$$
Since the integrand $p(x,y)P(y,B_x)$ is the product of two strictly positive terms for all $y \in B_x$, we conclude that $P^2(x,B_x) > 0$, showing that the chain is aperiodic.

\textit{Smallness of compacts:}
By Theorem 5.5.7 and Theorem 6.2.5 in \cite{meyn2012markov}, if a chain is $\pi$-irreducible and aperiodic, then every petite set is small, and if the chain is also a $T$-chain, then every compact set is petite. Therefore, for a  $\pi$-irreducible and aperiodic $T$-chain, every compact set is small.

(ii) By part (i) and Assumption \ref{ass:pi_q_regularities} the conditions of Proposition 3.1 in \cite{roberts1996geometric} are satisfied, which implies the result.
\end{proof}

\begin{proof}[Proof of Proposition \ref{prop: counterexample}]
First note that
\begin{align*}
   \lim_{|x|\to\infty} \int\alpha(x,y) Q(x,dy)
    \ge \lim_{|x|\to\infty}\int \min\left\{Q(x,dy),\frac{1}{2\sqrt{2\pi}}\exp\left(\frac{x^2}{2}+\frac{xy}{2}-\frac{9y^2}{8}\right)dy\right\}= 1,
\end{align*}
To verify that $Q$ is not geometrically ergodic, suppose for contradiction that it is. By Proposition \ref{prop: MH is a T-chain}, there exists $V\ge1$ such that 
\[
\limsup_{|x|\to\infty}\frac{QV(x)}{V(x)} = \limsup_{|x|\to\infty} \left[\frac{(1-\varepsilon(x))Q_1V(x)}{V(x)} + \varepsilon(x)\frac{Q_2V(x)}{V(x)}\right]<1
\]
Since $\varepsilon(x)Q_2V(x)/V(x)\ge 0$, this inequality implies $\limsup_{|x|\to\infty}Q_1V(x)/V(x)<1$.  This implies further that $Q_1V(x) \leq \rho V(x)$ for sufficiently large $|x|$ and some $\rho < 1$.  In Lemma \ref{lem:ar1_vgrowth} we show that any such $V(x) \geq c|x|^\eta$ for all sufficiently large $|x|$. 
Note also that
\begin{align*}
    \limsup_{|x|\to\infty}\frac{QV(x)}{V(x)} 
    &\ge \limsup_{|x|\to\infty} \left[\varepsilon(x)\frac{Q_2V(x)}{V(x)}\right].
\end{align*}
Using this and the condition $QV(x) \leq \rho V(x)$ for large enough $|x|$ implies
\[
\varepsilon(x)Q_2V(x) \leq QV(x) \leq \rho V(x),
\]
meaning we have the upper bound $Q_2V(x) \leq \rho(2+|x|)V(x)$.  
Applying the lower bound at any $y \in (\exp(x^2), \exp(x^2)+1)$ gives $V(y) \geq c\exp(\eta x^2)$. Integrating over the unit interval gives
\[
Q_2V(x) = \int_{\exp(x^2)}^{\exp(x^2)+1} V(y)\,dy \geq c\exp(\eta x^2).
\]
Combining with the upper bound leads to the inequality
\[
\rho(2+|x|)V(x) \geq c\exp(\eta x^2) \implies V(x) \geq \frac{c}{\rho(2+|x|)}\exp(\eta x^2).
\]
Now, applying this new lower bound at the point $y = \exp(x^2)$ gives
\[
V(\exp(x^2)) \geq \frac{c}{\rho(2+\exp(x^2))}\exp\bigl(\eta e^{2x^2}\bigr) \geq \exp\bigl(\eta' e^{2x^2}\bigr)
\]
for any $\eta' < \eta$ and large enough $|x|$.  This provides a lower bound for $Q_2V(x)$, leading to
\[
\exp(\eta' e^{2x^2}) \leq Q_2V(x) \leq \rho(2+|x|)V(x) \implies V(x) \geq \exp(\eta'' e^{2x^2})
\]
for $\eta''< \eta'$ in the same $|x|$ limit.  For any $V$ that satisfies the above, however, letting $Z \sim N(0,1)$ it holds that
\[
Q_1V(x) 
= 
\mathbb{E}[V(x/2+Z)] 
\geq 
\mathbb{E}\left[\exp\left(\eta''  e^{2(x/2+Z)^2}\right)\right] 
= 
\frac{1}{\sqrt{2\pi}}\int \exp\left(\eta'' e^{2(x/2+z)^2} - \frac{z^2}{2}\right)dz = \infty.
\]
This directly contradicts the requirement $Q_1V(x) \leq \rho V(x) < \infty$, from which we conclude that $Q$ cannot be geometrically ergodic.

To see that $P(x,dy)$ is geometrically ergodic, choosing $V(x) = \exp(x^2/4) $, note that
\begin{align*}
\limsup_{|x|\to\infty} \frac{PV(x)}{V(x)} & =\limsup_{|x|\to\infty}\left[\frac{\int V(y)\alpha(x,y)Q(x,dy) }{V(x)}+ \int(1-\alpha(x,y)) Q(x,dy) \right]\\
 &= \limsup_{|x|\to\infty}\frac{\int\alpha(x,y)(1-\varepsilon(x))V(y)Q_1(x,dy)  + \int_{\exp(x^2)}^{\exp(x^2)+1}\alpha(x,y)\varepsilon(x) V(y)}{V(x)}.
\end{align*}
Upon noting that $1-\varepsilon(x) > 1/2$ for large $|x|$ and $q_2(y,x) = 0$, then provided that $|x|$ is large enough
\begin{equation*}
    \alpha(x,y) \le \exp(x^2-y^2)\frac{q_1(y,x)}{0.5q_1(x,y)} = 2\exp(x^2-y^2) \frac{\exp\left(-(x-y/2)^2/2\right) }{\exp\left(-(y-x/2)^2/2\right)} = 2\exp\left(\frac{5}{8}x^2-\frac{5}{8}y^2\right)
\end{equation*}
for $x\notin (\exp(y^2),\exp(y^2)+1)$.  Noting that $V(x) = \exp(x^2/4)$ then this implies
\begin{equation*}
    \limsup_{|x|\to\infty} \frac{PV(x)}{V(x)}
    \le 
    \limsup_{|x|\to \infty}\left[(1-\varepsilon(x))\frac{Q_1V(x)}{V(x)} + 2\exp\left(\frac{5}{8}x^2-\frac{5}{8}\exp(2x^2)\right)\varepsilon(x)\frac{V(\exp(x^2)+1)}{V(x)}\right],
\end{equation*}
which upon simplification can be written
\begin{equation}
\label{eq:pv_expression}
    \limsup_{|x|\to\infty}\left[\left(1-\frac{1}{2+|x|}\right)\frac{Q_1V(x)}{V(x)} + \frac{2}{2+|x|} \exp\left(\frac{3}{8}x^2 - \frac{3}{8}\exp(2x^2)+\frac{1}{2}\exp(x^2)+\frac{1}{4}\right)\right].
\end{equation}
The second term on the right-hand side asymptotes to $0$ owing to the dominant $- 3\exp(2x^2)/8$ term inside the exponential.  Turning to the first we have
\begin{align*}
    \limsup_{|x|\to\infty}\frac{Q_1V(x)}{V(x)} 
    &= 
    \limsup_{|x|\to\infty}\frac{\mathbb{E}[\exp((x/2+Z)^2/4)]}{\exp(x^2/4)}
    = \limsup_{|x|\to\infty}\exp(-3x^2/16)\mathbb{E}\left[\exp\left(\frac{xZ}{4}+\frac{Z^2}{4}\right)\right].
\end{align*}
Expanding the expectation gives $\mathbb{E}\left[\exp\left(xZ/4+Z^2/4\right)\right] = c'\exp\left(x^2/16\right)$ for some $c'<\infty$, meaning
\begin{equation*}
    \lim_{|x|\to\infty}\frac{Q_1V(x)}{V(x)}  = \lim_{|x|\to\infty}c'\exp\left(-\frac{x^2}{8}\right)=0.
\end{equation*}
Substituting into \eqref{eq:pv_expression} therefore gives
\begin{equation*}
    \limsup_{|x|\to\infty} \frac{PV(x)}{V(x)}
    \le 
    \limsup_{|x|\to\infty}\frac{Q_1V(x)}{V(x)} = 0,
\end{equation*}
from which the result follows.
\end{proof}

\begin{lemma}
\label{lem:ar1_vgrowth}
If $Q_1(x,\cdot) = N(x/2, 1)$, and $V : \mathbb{R} \to [1, \infty)$ satisfies
\begin{equation}\label{eq:drift}
    Q_1 V(x) \leq \rho V(x), \qquad |x| > x_0
\end{equation}
for some $\rho \in (0,1)$ and $x_0 < \infty$, then there exist constants $c > 0$ and $\eta > 0$ such that $V(x) \geq c\, |x|^{\eta}$ for all sufficiently large $|x|$.
\end{lemma}

\begin{proof}
Let $C = [-x_0, x_0]$ and $\tau_C := \inf\{n \geq 1 : X_n \in C\}$. Set
$\varepsilon := -\log \rho > 0$.  Following the proof of Theorem~3.3 of \cite{mengersen1996rates}, it must be that any valid $V$ satisfies
\begin{equation}\label{eq:logV}
    \log V(x) \geq \varepsilon\mathbb{E}_x[\tau_C]
\end{equation}
for all sufficiently large $|x|$.\footnote{Note an inconsequential typographical error in equation (31) of \cite{mengersen1996rates}, in which the bound is reported as $\mathbb{E}_x[\tau_C]/\varepsilon$.} It therefore remains to obtain a lower bound on
$\mathbb{E}_x[\tau_C]$.  Set $n(x) := \lfloor (\log_2 x)/2 \rfloor$ for some $x > x_0$. Since
$\mathbb{P}_x(\tau_C > k)$ is non-increasing in $k$,
\[
    \mathbb{E}_x[\tau_C]
    = \sum_{k=0}^\infty \mathbb{P}_x(\tau_C > k)
    \geq n(x) \cdot \mathbb{P}_x\left(\tau_C > n(x)\right).
\]
Note that if $X_0 = x$ then $X_k \sim N(x/2^k,
\sigma_k^2)$ with $\sigma_k^2 < 4/3$. For each $k \leq n(x)$ the mean
of $X_k$ satisfies $x/2^k \geq x/2^{n(x)} \geq \sqrt{x}$, so using a Gaussian tail bound
\[
    \mathbb{P}_x(X_k \in C) \leq \Phi \left(\frac{x_0 - \sqrt{x}}{2/\sqrt{3}}\right)
    \leq e^{-c_0 x}
\]
for some constant $c_0 > 0$ and all sufficiently large $x$. A union bound over
the $n(x) = O(\log x)$ steps gives $\mathbb{P}_x(\tau_C \leq n(x)) \leq n(x)e^{-c_0 x}
\to 0$, so that $\mathbb{P}_x(\tau_C > n(x)) \geq 1/2$ for all sufficiently large $x$.
Therefore
\[
    \mathbb{E}_x[\tau_C] \geq \frac{n(x)}{2} \geq \frac{\log_2 x}{8} = \frac{\log(x)}{8\log(2)}.
\]
Substituting into \eqref{eq:logV} gives
\[
    \log V(x) \geq \log(x^\eta),
\]
with $\eta := \varepsilon/(8\log (2)) > 0$ and therefore $V(x) \geq x^{\eta}$.  An analogous argument holds when $x < -x_0$, which completes the proof.
\end{proof}

\section{Proofs of Section \ref{sec: RWM and GWM on 1D R}}
\label{app:proofs_1D_RWM_GWM}

\subsection{Proofs of Subsection \ref{subsec:heavy_tailed_target}}
\label{app:proofs_gwm_polynomialtails}


\begin{lemma}
\label{lem:lowerbound_polynomialtails}
    If $\pi$ satisfies Assumption \ref{ass:heavy-tailed}, then the polynomial convergence rate of the guided walk Metropolis is at least $r$.
\end{lemma}

\begin{proof}
We prove the result by constructing a drift condition of the form
\[
PV(x) \leq V(x) - \gamma V(x)^{\alpha} + b\mathbb{I}_C(x)
\]
for some Lyapunov function $V:\mathbb{R} \to [1,\infty)$, $0 < b,\gamma < \infty$, a small set $C$ and $0 < \alpha < 1$.  Polynomial ergodicity with rate at least $\alpha/(1-\alpha)$ in total variation distance then follows from Theorem 3.6 of \citep{jarner2002polynomial}.

Consider the Lyapunov function 
\begin{equation*}
    V(x,p) = \exp(\delta \operatorname{sgn}(x)p) \pi(x)^{-\beta}
\end{equation*}
where $\delta>0$ and $1/
(r+1)<\beta<1$. Let $C = [-x_0,x_0]$ be a small set, where $x_0$ is chosen sufficiently large such that for all $|x|>x_0$, the acceptance region is $A(x) = [-|x|,|x|]$ and the rejection region is $R(x) = (-\infty,-|x|)\cup (|x|,\infty)$. Since $Z$ is Gaussian, writing $q(z)$ to denote its density we have
\begin{equation*}
    \mu = \int_0^\infty z q(z) dz < \infty, \quad \sigma^2 = \int_0^\infty z^2 q(z) dz < \infty.
\end{equation*}
Note that the Lyapunov function $V(x,p)$ shows two different patterns in the exponential term, and we will discuss these two cases separately.

\textit{Case 1:} $\operatorname{sgn}(x)p > 0$.  Here $V(x,p) = e^\delta \pi(x)^{-\beta} = e^\delta |x|^{(r+1)\beta}$. The transition operator gives
\begin{equation}
    PV(x,p) = 2\int_0^\infty V(x+pz,p)\frac{\pi(x+pz)}{\pi(x)}q(z)dz + 2\int_0^\infty V(x,-p)\left(1-\frac{\pi(x+pz)}{\pi(x)} \right)q(z)dz
\end{equation}
Since $\operatorname{sgn}(x)p>0$, we have $\operatorname{sgn}(x+pz)p>0$ for all $z>0$. For the rejection probability, by Taylor's theorem with explicit remainder there exists $\xi \in (0, z/|x|)$ such that
\begin{equation*}
    1 - \left(1+\frac{z}{|x|}\right)^{-(1+r)} = \frac{(1+r)z}{|x|} - \frac{(1+r)(2+r)}{2}\frac{z^2}{|x|^2}\left(1+\xi\right)^{-(3+r)}.
\end{equation*}
Since $(1+\xi)^{-(3+r)} \leq 1$ for $\xi > 0$, the remainder is bounded in absolute value by $\frac{(1+r)(2+r)}{2}\frac{z^2}{|x|^2}$. Integrating term by term against $q(z)$ is therefore justified by the finite second moment $\sigma^2 < \infty$, giving
\begin{equation*}
    \int_0^\infty \left(1-\frac{\pi(x+pz)}{\pi(x)}\right)q(z)dz = \frac{(1+r)\mu}{|x|} - R(x),
\end{equation*}
where the remainder satisfies
\begin{equation*}
    0 \leq R(x) \leq \frac{(1+r)(2+r)\sigma^2}{2|x|^2}.
\end{equation*}
Thus,
\begin{equation}
\begin{aligned}
    \frac{PV(x,p)-V(x,p)}{V(x,p)} 
    &= 2\int_0^\infty \left[\left(\frac{\pi(x+pz)}{\pi(x)}\right)^{1-\beta}-1\right]q(z)dz + 2e^{-2\delta}\int_0^\infty \left(1-\frac{\pi(x+pz)}{\pi(x)}\right) q(z) dz\\
    &\le 2\int_0^\infty\left[\left(1+\frac{z}{|x|}\right)^{(r+1)(\beta-1)}-1\right]q(z)dz + 2e^{-2\delta}\left[\frac{(1+r)\mu}{|x|} - R(x)\right].
\end{aligned}
\end{equation}
Applying Taylor's theorem to the first integrand as before gives
\begin{align}
    \frac{PV(x,p)-V(x,p)}{V(x,p)}  &\le -\frac{2(r+1)(1-\beta)\mu}{|x|} + \frac{c'\sigma^2}{|x|^2} + \frac{2e^{-2\delta}(1+r)\mu}{|x|}\\
    &= \frac{2(r+1)\mu(\beta - 1 + e^{-2\delta})}{|x|} + \frac{c'\sigma^2}{|x|^2}.
\end{align}
Now choose $\delta$ large enough that $e^{-2\delta} < 1-\beta$, which is possible since $\beta < 1$, so that the coefficient $\beta - 1 + e^{-2\delta} < 0$. Setting $c_1 = 2(r+1)\mu(1-\beta-e^{-2\delta}) > 0$, and choosing $|x|$ large enough that $c'\sigma^2/|x| < c_1/2$, gives
\begin{align}
    \frac{PV(x,p)-V(x,p)}{V(x,p)} \leq -\frac{c}{|x|}
\end{align}
for some $c > 0$, as required.

\textit{Case 2:} $\operatorname{sgn}(x)p<0$.  Now $V(x,p) =e^{-\delta} \pi(x)^{-\beta} = e^{-\delta} |x|^{(r+1)\beta}$. The transition operator gives
\begin{equation}
\begin{aligned}
    PV(x,p) 
    &= 2\int_0^{|x|} V(x+pz,p) q(z) dz + 2\int_{|x|}^{2|x|} V(x+pz,p) q(z)dz\\
    &\quad + 2\int_{2|x|}^\infty V(x+pz,p)\frac{\pi(x+pz)}{\pi(x)}q(z)dz \\
    &\quad + 2\int_{2|x|}^\infty V(x,-p) \left(1-\frac{\pi(x+pz)}{\pi(x)}\right)q(z)dz.
\end{aligned}
\end{equation}
For $z\in(0,|x|)$, we have $\operatorname{sgn}(x+pz)p <0$, so $V(x+pz,p) = e^{-\delta} |x+pz|^{(r+1)\beta}$. For $z \in (|x|,2|x|)$, we have $\operatorname{sgn}(x+pz)p>0$, so $V(x+pz,p) = e^\delta |x+pz|^{(r+1)\beta}$. Therefore
\begin{equation}
\begin{aligned}
    \frac{PV(x,p) - V(x,p)}{V(x,p)}
    &= 2\int_0^{|x|} \left[\left|\frac{x+pz}{x}\right|^{(r+1)\beta}-1\right]q(z) dz + 2e^{2\delta} \int_{|x|}^{2|x|} \left|\frac{x+pz}{x}\right|^{(r+1)\beta}q(z)dz\\
    &\quad + 2e^{2\delta}\int_{2|x|}^\infty \left|\frac{x+pz}{x}\right|^{(r+1)\beta}\frac{\pi(x+pz)}{\pi(x)}q(z)dz\\
    &\quad + 2e^{2\delta}\int_{2|x|}^\infty \left(1- \frac{\pi(x+pz)}{\pi(x)}\right)q(z) dz -2\int_{|x|}^\infty q(z)dz\\
    &\le 2\int_0^{|x|}\left[\left|\frac{x+pz}{x}\right|^{(r+1)\beta}-1\right]q(z) dz + 2(e^{2\delta}-1) \mathbb{P}(Z > |x|).
\end{aligned}
\end{equation}
Applying Taylor's theorem with explicit remainder to the first integrand, for $z \in (0,|x|)$ and $\operatorname{sgn}(x)p < 0$ there exists $\xi \in (0, z/|x|)$ such that
\begin{equation*}
    \left(1-\frac{z}{|x|}\right)^{(r+1)\beta} - 1 = -(r+1)\beta\frac{z}{|x|} + \frac{(r+1)\beta((r+1)\beta-1)}{2}\frac{z^2}{|x|^2}(1-\xi)^{(r+1)\beta - 2}.
\end{equation*}
Since $(1-\xi)^{(r+1)\beta-2} \leq 1$ for $\xi \in (0,1)$ and $(r+1)\beta > 1$ by assumption, the remainder is bounded in absolute value by $c'' z^2/|x|^2$ where $c'' = (r+1)\beta((r+1)\beta-1)/2 > 0$. Integrating term by term against $q(z)$ over $(0,|x|)$ is therefore justified by $\sigma^2 < \infty$, giving
\begin{align}
    \frac{PV(x,p) - V(x,p)}{V(x,p)} &\le -\frac{2(r+1)\beta}{|x|}\int_0^{|x|}zq(z)dz + \frac{c''\sigma^2}{|x|^2} + 2(e^{2\delta}-1)\mathbb{P}(Z>|x|).
\end{align}
Since $q(z)$ is Gaussian, $\mathbb{P}(Z > |x|)$ decays super-exponentially in $|x|$, so for any fixed $\delta$ we have $2(e^{2\delta}-1)\mathbb{P}(Z>|x|) = o(1/|x|)$ as $|x| \to \infty$. Moreover $\int_0^{|x|}zq(z)dz \to \mu$ as $|x|\to\infty$, so for sufficiently large $|x|$ the dominant term is $-2(r+1)\beta\mu/|x|$, and the same conclusion as in case 1 holds, namely
\begin{align}
    \frac{PV(x,p) - V(x,p)}{V(x,p)} \leq -\frac{c}{|x|}
\end{align}
for some $c > 0$.

Combining both cases shows that outside $C$ there exists $c^*>0$ such that the polynomial drift condition
\begin{equation}
    PV(x,p) - V(x,p) \le -c^*V(x,p)^{((r+1)\beta-1)/[(r+1)\beta]}
\end{equation}
holds. The polynomial rate is therefore at least
\begin{equation}
    \frac{((r+1)\beta-1)/[(r+1)\beta]}{1-((r+1)\beta-1)/[(r+1)\beta]} = (r+1)\beta -1.
\end{equation}
For any $\varepsilon > 0$, choosing $\beta = 1-\varepsilon \in (1/(r+1), 1)$ gives a valid Lyapunov function and a polynomial rate of at least $(r+1)(1-\varepsilon) - 1 = r - (r+1)\varepsilon$. Since this holds for all $\varepsilon > 0$, we conclude that the polynomial rate is at least $\sup_{\varepsilon > 0}\{r - (r+1)\varepsilon\} = r$, which completes the proof.
\end{proof}


\begin{lemma}
\label{lem:upperbound_polynomialtails}
    If $\pi$ satisfies Assumption \ref{ass:heavy-tailed}, then the polynomial convergence rate of the guided walk Metropolis with finite variance proposal is at most $r$.
\end{lemma}

\begin{proof}
Suppose that the algorithm begins from $(x,p) = (0,1)$. For any $n>0$, we can choose $k>0$ such that $kn>K$. Define the set $A_n = \{x:x>kn\}$, and let $\mu$ be the joint invariant measure of $X,P$. For notational simplicity we write $P^n((x,p),A)$ for $P^n((x,p),(A,\{+1,-1\}))$, and note that $\mu((A,\{+1,-1\})) = \pi(A)$.  By the definition of total variation distance, we have
\begin{align}
    \|P^n((0,1),\cdot) - \mu\|_{TV} &\ge |P^n((0,1),A_n) - \pi(A_n)|.
\end{align}    
To bound $P^n((0,1),A_n)$, we observe that if the chain is never rejected in $n$ steps, it continuously moves in the positive direction. By a stochastic comparison argument the probability that the chain reaches $A_n$ is at most the probability of reaching $A_n$ without any rejections. Therefore,
\begin{equation*}
    P^n((0,1),A_n) \le \mathbb{P}\left(\sum_{i=1}^n |z_i|>kn\right), \quad \text{where } z_i \stackrel{\text{i.i.d.}}{\sim} N(0,\sigma^2) \text{for} \, i=1, \cdots ,n.
\end{equation*}
Since each $|z_i|$ follows a half-normal distribution with tail decay at least as fast as a Gaussian, $|z_i|$ is sub-Gaussian. Let $\mu = \mathbb{E}[|z_i|] = \sigma\sqrt{2/\pi}$ and $\sigma_H^2 = \text{Var}(|z_i|)$. Then for all $\lambda>0$,
\begin{equation*}
    \mathbb{E}\left[e^{\lambda(\sum_{i=1}^n |z_i| - n\mu)}\right]\le \exp\left(\frac{n\lambda^2\sigma_H^2}{2}\right).
\end{equation*}
Applying Markov's inequality to the exponential function, we obtain for all $\lambda>0$
\begin{equation}
\begin{aligned}
    \mathbb{P}\left(\sum_{i=1}^n|z_i| >kn\right) &= \mathbb{P}\left(\sum_{i=1}^n|z_i| - n\mu>(k-\mu)n\right)\\
    &= \mathbb{P}\left(\exp\left(\lambda \sum_{i=1}^n(|z_i| - \mu)\right) >\exp(\lambda(k-\mu)n)\right)\\
    &\le \frac{\mathbb{E}\left[\exp\left(\lambda \sum_{i=1}^n(|z_i| - \mu)\right)\right]}{\exp(\lambda(k-\mu)n)}\\
    &\le \exp\left(\frac{n\lambda^2\sigma_H^2}{2} - \lambda n (k-\mu)\right).
\end{aligned}
\end{equation}
Since this holds for all $\lambda>0$, we can minimize over $\lambda$ by choosing $\lambda^* = (k-\mu)/\sigma_H^2$, giving
\begin{align}
    P^n((0,1),A_n) \le  \mathbb{P}\left(\sum_{i=1}^n|z_i| >kn\right)\le \exp\left(-\frac{n(k-\mu)^2}{2\sigma_H^2}\right).
\end{align}
Since $kn>K$ by construction, for $x>kn$ we have $\pi(x) = C_0 x^{-(1+r)}$. Combining gives
\begin{equation}
    \pi(A_n) = \int_{kn}^\infty \frac{C_0}{x^{1+r}}  dx = \frac{C_0}{r(kn)^r} = \frac{C'}{n^r}, \quad \text{where } C' = C_0/(rk^r),
\end{equation}  
which implies that for all $n\ge n_0$, for suitably large $n_0$ and constant $k>0$
\begin{equation}
    |\pi(A_n) - P^n((0,1),A_n)| \ge \left|\frac{C'}{n^r} - \exp\left(-\frac{n(k-\mu)^2}{2\sigma_H^2}\right)\right| \ge \frac{c}{n^r}
\end{equation}
where $c>0$ is a constant. The last inequality holds because for any fixed $r>0$, we can choose large $k=k(r)$, and $C' = C_0/(rk^r):= C'(r)$ such that both relate to $r$ and $\exp(-n(k-\mu)^2/(2\sigma_H^2)) < C'/(2n^r)$ for all $n\ge n_0$ for suitably large $n_0$, which implies the bound with $c = C'/2$.  We have therefore established that
\begin{equation}
    \|P^n((0,1),\cdot) - \pi(\cdot)\|_{TV} \ge \frac{c}{n^r}, \quad n\ge n_0,
\end{equation}
from which the result follows.
\end{proof}

\subsection{Proofs of Subsection \ref{subsec:light_tailed_target}}
\label{app:proofs_gwm_lighttails}

\begin{lemma}
\label{lem:convergence_alpha_light_tails}
If $\pi$ satisfies Assumption \ref{ass:strict_convexity}, then the (Gaussian) random walk Metropolis satisfies
\begin{equation*}
    \lim_{x\to \infty}\int_{x}^{\infty} \alpha(x,y)Q(x,dy) = \lim_{x\to -\infty}\int_{-\infty}^{x} \alpha(x,y)Q(x,dy) = 0.
\end{equation*}
\end{lemma}

\begin{proof}
Considering $[x, \infty)$ for $x>0$
\begin{equation*}
\begin{aligned}
    \int_x^\infty \alpha(x,y)Q(x,dy) = \int_{0}^{\infty} \alpha(x, x+\epsilon z)\phi(z) dz
    = 
    \int_0^{\infty} \min(1, \exp(U(x) - U(x+\epsilon z)))\phi(z)dz
\end{aligned}
\end{equation*}
where $\phi(z)$ is the standard Normal density. The mean value theorem states that
\begin{equation*}
    U(x+\epsilon z) - U(x) \geq \epsilon z \inf_{t \in [x, x+\epsilon z]} U'(t) = \epsilon z U'(x)
\end{equation*}
where the last equality follows from monotonicity of $U'$.  By assumption $U(x+\epsilon z) -  U(x) \to \infty$ as $x \to \infty$ because $U'(t) \to \infty$, giving the point-wise limit
\begin{equation*}
    \lim_{x\to \infty}\alpha(x, x+\epsilon z) = 0
\end{equation*}
for any fixed $z > 0$. Since $\alpha(x,x+\epsilon z) \leq 1$ for any $z$, then the bounded convergence theorem can be used to conclude.  For the tail region $(-\infty, x]$ when $x \to -\infty$ an identical argument can be applied, which completes the proof.
\end{proof}

\end{document}